\documentclass[12pt]{amsart}
\usepackage{amssymb}
\usepackage{color}
\newtheorem{claim}{}[section]
\newtheorem{theorem}[claim]{Theorem}
\newtheorem{lemma}[claim]{Lemma}
\newtheorem{proposition}[claim]{Proposition}
\newtheorem{corollary}[claim]{Corollary}
\renewenvironment{proof}{\noindent{\it Proof. \hskip0pt}}
                      {$\square$\par\medskip}

\begin{document}
\baselineskip 6.2 truemm
\parindent 1.5 true pc

\newcommand\lan{\langle}
\newcommand\ran{\rangle}
\newcommand\tr{{\text{\rm Tr}}\,}
\newcommand\ot{\otimes}
\newcommand\wt{\widetilde}
\newcommand\join{\vee}
\newcommand\meet{\wedge}
\renewcommand\ker{{\text{\rm Ker}}\,}
\newcommand\im{{\text{\rm Im}}\,}
\newcommand\mc{\mathcal}
\newcommand\transpose{{\text{\rm t}}}
\newcommand\FP{{\mathcal F}({\mathcal P}_n)}
\newcommand\ol{\overline}
\newcommand\JF{{\mathcal J}_{\mathcal F}}
\newcommand\FPtwo{{\mathcal F}({\mathcal P}_2)}
\newcommand\hada{\circledcirc}
\newcommand\id{{\text{\rm id}}}
\newcommand\tp{{\text{\rm tp}}}
\newcommand\pr{\prime}
\newcommand\e{\epsilon}
\newcommand\inte{{\text{\rm int}}\,}
\newcommand\ttt{{\text{\rm t}}}
\newcommand\spa{{\text{\rm span}}\,}
\newcommand\conv{{\text{\rm conv}}\,}
\newcommand\rank{\ {\text{\rm rank of}}\ }
\newcommand\vvv{\mathbb V_{m\meet n}\cap\mathbb V^{m\meet n}}
\newcommand\ppp{\mathbb P_{m\meet n} + \mathbb P^{m\meet n}}
\newcommand\re{{\text{\rm Re}}\,}
\newcommand\la{\lambda}
\newcommand\msp{\hskip 2pt}
\newcommand\ppt{\mathbb T}
\newcommand\rk{{\text{\rm rank}}\,}

\title{Exposed faces for decomposable positive linear maps arising from completely positive maps}
\author{Hyun-Suk Choi}
\address{Department of Mathematics\\Seoul National University\\Seoul 151-742, Korea}
\email{spgrass2@snu.ac.kr}

\author{Seung-Hyeok Kye}
\address{Department of Mathematics and Institute of Mathematics\\Seoul National University\\Seoul 151-742, Korea}
\email{kye@snu.ac.kr}
\thanks{This work was partially supported by PARC}



\begin{abstract}
Let $D$ be a space of $2\times n$ matrices. Then the face of the cone of all completely positive maps from $M_2$
into $M_n$ given by $D$ is an exposed face of the bigger cone of all decomposable positive linear maps if and only if the set of all
rank one matrices in $D$ forms a subspace of $D$ together with zero and $D^\perp$ is spanned by rank one
matrices.
\end{abstract}

\maketitle

%%%%%%%%%%%%%%%%%%%%%%%%%%%%%%%%%%%%%%%%%%%%%%%%%%%%%%%%%%
\section{Introduction}

Let $M_n$ be the $C^*$-algebra of all $n\times n$ matrices over
the complex field, and $\mathbb P_1[M_m,M_n]$ the cone of all
positive linear maps from $M_m$ into $M_n$. The cone $\mathbb P_1$
is very important in the recent development of entanglement theory
in quantum physics, and plays crucial role to distinguish
entangled states from separable ones. Nevertheless, the whole
convex structures of the cone $\mathbb P_1$ is extremely
complicated, and far from being completely understood. In the case
of $m=n=2$, all extreme points of the convex set of all unital
positive linear maps were found by St\o rmer \cite{stormer}, and
the whole facial structures of the cone $\mathbb P_1$ has been
characterized in \cite{byeon-kye}. See also \cite{kye-2by2_II}. On
the other hand, Yopp and Hill \cite{yopp} showed that the
elementary maps
$$
\phi_V: X\mapsto V^*XV,
\qquad
\phi^V: X\mapsto V^*X^\ttt V,
$$
where $V$ is an $m\times n$ matrix, generate extremal rays in the
cone $\mathbb P_1$, and they are exposed if the rank of $V$ is one
or full. Recently, Marciniak \cite{marcin_exp} showed that those
maps generate exposed extremal rays, in general.

The cone $\mathbb P_1$ has the subcone, denoted by $\mathbb D$,
consisting of all decomposable positive linear maps which are, by
definition, the sums of completely positive linear maps and
completely copositive linear maps. In the case of $m=2$, it was shown by
Woronowicz \cite{woronowicz} that $\mathbb P_1$
coincides with $\mathbb D$ if and only if $n\le 3$.
Note that the map $\phi_V$ generates an
extremal ray of the cone $\mathbb C\mathbb P$ of all completely
positive linear maps, and every map which generates an extremal ray
of the cone $\mathbb C\mathbb P$ is in this form. Therefore, the
above mentioned results tell us that every extremal ray of the cone
$\mathbb C\mathbb P$ is an exposed extremal ray of the much bigger cone
$\mathbb P_1$.

Recall that every completely positive linear map from $M_m$ into $M_n$ is of the form
$$
\phi_{\mathcal V}=\sum_{V\in\mathcal V}\phi_V,
$$
for a subset $\mathcal V$ of $m\times n$ matrices. We also recall
\cite{kye-cambridge} that every face of the cone $\mathbb C\mathbb P$ is of the form
$$
\Phi_D
=\{\phi_{\mathcal V}: \mathcal V\subset D\},
$$
for a subspace $D$ of the inner product space $M_{m\times n}$ of
all $m\times n$ matrices. The result of Marciniak
\cite{marcin_exp} says that $\Phi_D$ is an exposed face of
$\mathbb P_1$ whenever $\dim D=1$. If $V$ is of rank one, then
$\phi_V$ is both completely positive and completely copositive. On
the other hand, if rank of $V$ is greater than one, then $\phi_V$
is not completely copositive. We note that $V^\perp$ is spanned by rank one
matrices in both cases.

Every face of the cone $\mathbb D$ is determined by a pair $(D,E)$ of subspaces of $M_{m\times n}$, and it is
exposed by separable states if and only if there are rank one matrices $\xi_\iota \eta_\iota^*$ in $M_{m\times n}$
such that
\begin{equation}\label{range}
D^\perp=\spa \{\xi_\iota \eta_\iota^*\},\qquad
E^\perp=\spa \{\bar \xi_\iota \eta_\iota^*\},
\end{equation}
where $\bar \xi$ denotes the vector whose entries are conjugate of
the corresponding entries of the vector $\xi$. See the next section
for the details. This condition also arises in the context of the
range criterion for separability \cite{p-horo} and characterization
of some faces of the cone generated by separable states
\cite{choi_kye}. We say that a pair $(D^\perp, E^\perp)$ {\sl
satisfies the range criterion} if there are rank one matrices
$\xi_\iota \eta_\iota^*$ satisfying the condition (\ref{range}).

In general, it is very difficult to determine if a given pair of
subspaces satisfies the range criterion. One of the important step in \cite {marcin_exp} is to prove that
$(D^\perp, \{0\}^\perp)$ satisfies the range criterion if $D$ is the one dimensional subspace generated by
a matrix whose rank is at least two. In the case of $m=2$, Augusiak, Tura and Lewenstein \cite{aug} recently showed that
$(D^\perp, \{0\}^\perp)$ satisfies the range criterion whenever $D$ is completely entangled, that is, there are no rank one
matrices in $D$.

Note that this is not the case for $m=3$, since there exists a
$4$-dimensional completely entangled subspace of $M_{3\times 3}$
whose orthogonal complement has only six rank one matrices up to
constant multiples. See  \cite{ha+kye}. This is the {\sl generic} case for
$4$-dimensional subspaces of $M_{3\times 3}$ by \cite{walgate}. See
also \cite{bdmsst} for $4$-dimensional completely entangled
subspaces of $M_{3\times 3}$ whose orthogonal complements are spanned by
orthogonal rank one matrices. We refer to recent papers \cite{chen} and \cite{sko} for detailed studies
of $4$-dimensional completely entangled subspaces of $M_{3\times 3}$.

We also note that there exists a
completely entangled subspace of $M_{4\times 4}$ whose orthogonal
complement is not spanned by rank one matrices. See
\cite{bhat}. More recently, completely entangled subspaces of $M_{3\times 4}$
whose orthogonal complements are also completely entangled have been produced by
numerical searches in \cite{lein}.

In this note, we restrict our attention to the case of $m=2$, and we look for conditions of a subspace $D$ of $M_{2\times n}$
for which the face $\Phi_D$ of $\mathbb C\mathbb P$ becomes an exposed face of the bigger cone $\mathbb D$.
First of all, we show that if this is the case, then the set of all matrices in $D$ whose ranks are one or zero forms
a subspace of $D$. We say that a subspace of matrices is {\sl completely separable} if it consists of rank one or zero matrices.
Note that the range space of a state in the block matrices is completely separable then it is always separable,
here we identify $M_{m\times n}$ and $\mathbb C^n\otimes \mathbb C^m$.
We show that $\Phi_D$ is an exposed face of $\mathbb D$ if and only if the following two conditions are satisfied:
\begin{enumerate}
\item[(i)]
The set of all matrices in $D$ whose ranks are one or zero forms a subspace.
\item[(ii)]
$D^\perp$ is spanned by rank one matrices.
\end{enumerate}
As a byproduct, we see that when $E$ is completely separable, a pair $(D^\perp,E^\perp)$ of subspaces of $M_{2\times n}$
satisfies the range criterion if and only if every rank one matrix in $D$ is a partial conjugate of a rank one matrix in $E$
and $D^\perp$ is spanned by rank one matrices. This extends the above mentioned result in \cite{aug}.

Note that the vector space
$M_{m\times n}$ is inner product space isomorphic to $\mathbb C^n\otimes \mathbb C^m$ by the correspondence
$$
[z_{ij}]\mapsto \sum_{i=1}^m\left(\sum_{j=1}^nz_{ij}e_j\right)\otimes e_i,
$$
where $\{e_i\}$ and $\{e_j\}$ are the standard orthonormal bases of $\mathbb C^m$ and $\mathbb C^n$, respectively.
Then the rank one matrix $\xi\eta^*$ in $M_{m\times n}$ corresponds to $\bar\eta\otimes \xi\in \mathbb C^n\otimes \mathbb C^m$:
$$
\xi\eta^* \leftrightarrow \bar\eta\otimes \xi,
$$
where the latter is called a product vector in quantum physics. We usually use the tensor
notations in this note, with few exceptions.

In the next section, we show that if $\Phi_D$ is a face of
$\mathbb D$ then the set of all product vectors in $D$ forms a
subspace, and investigate the structures of such subspaces.
Especially, dimensions of such spaces are less than or equal to
$n$. We also find conditions for which $D^\perp$ is spanned by
product vectors in this situation. In the Section 3, we show the
main theorem mentioned above, and close the note with remarks on
examples of spaces $D$ for which $\Phi_D$ become unexposed faces
of $\mathbb D$.

The authors are grateful to the authors of \cite{lein} for valuable discussion on that paper, and Professor
Young-Hoon Kiem for valuable discussion on product vectors. The most part of this work was done
when the second author was visiting Jeju University. He express his gratitude for the hospitality of the faculties
of the university.

\section{completely separable and completely entangled subspaces}

Every face of the cone $\mathbb D$ is determined \cite{kye_decom} by a pair $(D,E)$ of subspaces of $M_{m\times n}$,
and it is of the form
$$
\sigma(D,E)=\conv \{\Phi_D,\Phi^E\},
$$
where $\Phi^E=\{\phi^{\mathcal W}:\mathcal W\subset E\}$ and
$\phi^{\mathcal W}=\sum_{W\in\mathcal W}\phi^W$. This pair is uniquely
determined if we impose the condition
$$
\sigma(D,E)\cap \mathbb C\mathbb P=\Phi_D,\qquad
\sigma(D,E)\cap \mathbb C\mathbb C\mathbb P=\Phi^E,
$$
where $\mathbb C\mathbb C\mathbb P$ denotes the cone of all completely copositive linear maps.

\begin{proposition}\label{face}
Let $D$ be a subspace of $M_{m\times n}$. If $\Phi_D$ is a face of $\mathbb D$, then
the set $D_1$ of all product vectors in $D$ forms a
subspace of $D$.
\end{proposition}

\begin{proof}
Take a subspace
$E$ of $M_{m\times n}$ such that $\Phi_D=\sigma(D,E)$. Then we have
\begin{equation}\label{gjjk}
\Phi_D\cap \mathbb C\mathbb C\mathbb P=\Phi^E.
\end{equation}
Note that $\phi^V$ is completely positive if and only if $V$ is of rank one or zero.
Therefore, the above condition (\ref{gjjk}) tells us that $E$ is completely separable.
If $xy^*\in D_1$, then $\phi^{\bar x y^*}=\phi_{xy^*}$ belongs to $\Phi^E$,
and so $\bar xy^*\in E$. Conversely, if $\bar xy^*\in E$ then
$xy^*\in D_1$. Therefore, we see that
$$
D_1=\{xy^*: \bar xy^*\in E\}
$$
is a subspace of $D$, which is completely separable.
\end{proof}

Geometrically, this means that if $\Phi_D$ is a face of $\mathbb D$,
then $\Phi_D \cap \mathbb{CCP}$ is a face of both $\Phi_D$ and
$\mathbb {CCP}$. If $D$ and $D_1$ are as in Proposition \ref{face},
then the orthogonal complement of $D_1$ in $D$ is completely
entangled. Therefore we see that if $\Phi_D$ is a face of $\mathbb
D$, then $D$ is the direct sum of a completely separable subspace
and a completely entangled subspace. We note that the converse is
not true. To see this, let $D_1$ and $D_2$ be one dimensional
subspaces of $M_{2\times 2}$ spanned by $e_{11}$ and
$e_{12}+e_{21}+e_{22}$, respectively, where $\{e_{ij}\}$ denotes the standard matrix units. Then $D_1$ is completely
separable and $D_2$ is completely entangled. But $D_1\oplus D_2$
does not satisfy the condition in Proposition \ref{face}.

From now on, we pay attention to the case of $m=2$ and suppose that
$D$ is a subspace of $M_{2\times n}$ such that $\Phi_D$ is a face of
$\mathbb D$, so the set $D_1$ consisting of all product vectors in $D$ forms a completely separable subspace of $D$.

It is easy to see that every completely separable subspace of $\mathbb C^n\otimes \mathbb C^m=M_{m\times n}$ is of the form
$$
\beta\otimes A:=\{\beta\otimes\alpha: \alpha\in A\}
$$
for a vector $\beta\in\mathbb C^n$ and a subspace $A$ of $\mathbb C^m$, or of the form
$$
B\otimes\alpha:=\{\beta\otimes\alpha:\beta\in B\}
$$
for a vector $\alpha\in\mathbb C^m$ and a subspace $B$ of $\mathbb C^n$. In the case of $m=2$, it suffices to consider the following
cases;
$$
\beta\otimes\mathbb C^2,\qquad B\otimes \alpha,
$$
for $\alpha\in\mathbb C^2$, $\beta\in\mathbb C^n$ and a subspace $B$ of $\mathbb C^n$.
Since the case of $\beta\otimes \mathbb C^2$ is easy, we will concentrate on the case of $B\otimes\alpha$.
From now on, we fix a nonzero vector $\alpha^\perp\in\mathbb C^2$ which is orthogonal to the vector $\alpha$.
Note that every element
$a$ in the orthogonal complement of $B\otimes\alpha$ is uniquely expressed by
\begin{equation}\label{expression}
a=\gamma\otimes \alpha +(\beta +\delta)\otimes \alpha^\perp,
\end{equation}
where $\beta\in B$ and $\gamma, \delta\in B^\perp$.

\begin{proposition}\label{ro2n}
Let $D$ be a subspace of $M_{2\times n}$. Suppose that the set $D_1$
of all product vectors in $D$ forms a subspace of the
form $B\otimes \alpha$. We denote by $p$ the projection of $M_n$
onto $B$, and $D_2$ the orthogonal complement of $D_1$ in $D$. Then we have the following:
\begin{enumerate}
\item[(i)]
$\dim D\le n$.
\item[(ii)]
Every product vector in $D_2(1-p)$ is of the form $\gamma
\otimes \alpha$ for $\gamma\in B^\perp$.
\item[(iii)]
The maximum number of linearly independent product vectors in $D_2(1-p)$ is $\min (k, n-k)$, where $k=\dim D_1=\dim B$.
\end{enumerate}
\end{proposition}

\begin{proof}
(i).
Let $\{a_i:i=1,2,\dots,\ell\}$ be a linearly independent set of $D_2$, and write
$$
a_i=\gamma_i\otimes \alpha +(\beta_i +\delta_i)\otimes \alpha^\perp
$$
as in (\ref{expression}). Assume that $\ell>n-k=\dim B^\perp$ then $\{\gamma_i\}\subset B^\perp$ is linearly dependent, and so
$\sum_i c_i\gamma_i=0$ for a nonzero $(c_1,\dots,c_\ell)$. Then $\Sigma_i c_ia_i=\Sigma_i c_i(\beta_i +\delta_i)\otimes \alpha^\perp$
is a nonzero product vector in $D_2$. This contradiction shows that $\ell\le n-k$, and so we have $\dim D\le n$.

(ii). Note that every element of $D_2(1-p)$ is of the form
$$
a(1-p)=\gamma\otimes \alpha + \delta \otimes \alpha^\perp,
$$
for $a\in D_2$ as in (\ref{expression}).
Suppose that this is a product vector, and assume that $\delta \neq 0$. Then we see
$\gamma= s \delta$ for a constant $s \in \mathbb C$. Since $\beta
\otimes \alpha \in D_1$, we see that
$$
\begin{aligned}
\beta\otimes s\alpha+a
&= \beta\otimes s\alpha+ \delta\otimes s\alpha+ \beta\otimes\alpha^\perp+ \delta\otimes\alpha^\perp\\
&= (\beta+\delta)\otimes s\alpha+ (\beta+\delta)\otimes\alpha^\perp
=(\beta+\delta)\otimes(s\alpha+ \alpha^\perp)
\end{aligned}
$$
is a product vector in $D$, which does not belong to $D_1$. This contradiction shows that $\delta =0$.

(iii).
In the above argument, we have shown that if $a(1-p)\in D_2(1-p)$ is a product vector then $a$ is of the form
$a=\gamma\otimes \alpha +\beta \otimes \alpha^\perp$. Suppose that
$\{a_i(1-p) : i=1,\dots, \ell\}$ is linearly independent with
$a_i=\gamma_i \otimes \alpha +\beta_i \otimes \alpha^\perp$. Then it
is clear that $\{\gamma_i\}$ is linearly independent.
If $\{\beta_i\}$ is linearly dependent, then we see that $\sum_i c_i a_i =\sum c_i \gamma_i \otimes \alpha$ is a
product vector for a nonzero $(c_1,\dots, c_\ell)$, which is not contained in $D_1$. This contradiction shows that
$\{\beta_i\}$ is also linearly independent. Therefore, we have
$\ell \le \min (k, n-k)$
\end{proof}

Now, we will concentrate on the case when $D_2(1-p)$ has no nonzero
product vectors, and show that this happens if and only if $D^\perp$
is spanned by product vectors.

\begin{proposition}\label{decompo}
Let $D$, $D_1$, $D_2$, $p$ and $k$ be as was given in Proposition \ref {ro2n}, and
$\{a_j:j=1,2,\dots,\ell\}$ be a basis of $D_2$, where
$$
a_i=\gamma_i\otimes \alpha +(\beta_i +\delta_i)\otimes \alpha^\perp
$$
as was given in {\rm (\ref{expression})}. Then the following are equivalent:
\begin{enumerate}
\item[(i)]
$D^\perp$ is spanned by product vectors.
\item[(ii)]
$\{\delta_i:i=1,\dots,\ell\}$ is linearly independent.
\item[(iii)]
The subspace $D_2(1-p)$ in $M_{2\times n}$ is completely entangled.
\end{enumerate}
If $D$ satisfies the above conditions, then $\dim D\le n-1$.
\end{proposition}

\begin{proof}
First of all, we note that all of the following sets
$$
\{\gamma_i\},\qquad
\{\beta_i +\delta_i\},\qquad
\{ \gamma_i\otimes\alpha+ \delta_i\otimes\alpha^\perp\}
$$
are linearly independent, because $D_2$ has no product vectors.

If a product vector $\eta\otimes\xi$ belongs to $D^\perp$, then $\eta\otimes\xi\in D_1^\perp$ implies that
$\xi\perp \alpha$ or $\eta\in B^\perp$.
In the first case, the product vectors are of the forms $\eta\otimes\alpha^\perp$ with $\eta\in\mathbb C^n$. We denote by
$H_1$ the span of all those product vectors in $D^\perp$. We also denote by $H_2$ the span of all
product vectors $\eta\otimes\xi$
in $D^\perp$ with $\xi\in\mathbb C^2$ and $\eta\in B^\perp$.

Note that $\eta\otimes\alpha^\perp\in H_1$ belongs to $D^\perp$ if and only if it is orthogonal to $D_2$ if and only if
$\eta\perp \beta_i+\delta_i$ for each $i=1,2,\dots,\ell$. Since $\{\beta_i+\delta_i\}$ is linearly independent, we see that
$$
\dim H_1=n-\ell.
$$
We also note that an element in $H_2\subset M_{2\times n}(1-p)$ belongs to $D_2^\perp$ if and only if
it is orthogonal to $\gamma_i\otimes \alpha+ \delta_i\otimes\alpha^\perp$ for each $i=1,2,\dots,\ell$. Since
$\{ \gamma_i\otimes\alpha+ \delta_i\otimes\alpha^\perp\}$ is linearly independent, we see that
\begin{equation}\label{dim_2}
\dim H_2\le  2(n-k)-\ell.
\end{equation}
On the other hand, since every element in $H_1\cap H_2$ is of the form $\eta\otimes\alpha^\perp$ with $\eta\in B^\perp$,
it belongs to $D^\perp$ if and only if $\eta\perp\delta_i$ for each $i=1,2,\dots\ell$.
So, we see that
\begin{equation}\label{dim-intersection}
\dim(H_1 \cap H_2)=(n-k)- \dim\spa \{\delta_i:i=1,2,\dots,\ell\}.
\end{equation}
Therefore, we have
\begin{equation}\label{dim}
\begin{aligned}
\dim\spa\{H_1,H_2\}
&=\dim H_1+\dim H_2-\dim (H_1\cap H_2)\\
&\le 2n-k-2\ell+\dim\spa\{\delta_i\}.
\end{aligned}
\end{equation}

Now, $D^\perp$ is spanned by product vectors if and only if $D^\perp=\spa\{H_1, H_2\}$ if and only if
$$
\dim\spa\{H_1,H_2\}=\dim D^\perp=2n-k-\ell.
$$
By the inequality (\ref{dim}), this implies
$$
\dim\spa\{\delta_i\}\ge\ell,
$$
and we see that $\{\delta_i\}$ is linear independent. This shows the direction (i) $\implies$ (ii).

For the direction (ii) $\implies$ (iii),
we assume that $D_2(1-p)$ is not completely entangled, and there is an element $a=\sum_{i=1}^\ell c_ia_i$ in $D_2$
such that $a(1-p)$ is a product vector. Put
$$
\beta=\sum_{i=1}^\ell c_i\beta_i,\qquad
\gamma=\sum_{i=1}^\ell c_i\gamma_i,\qquad
\delta=\sum_{i=1}^\ell c_i\delta_i.
$$
Then, we have
$a= \gamma\otimes\alpha+ \beta\otimes\alpha^\perp+ \delta\otimes\alpha^\perp$.
By Proposition \ref{ro2n} (ii), we have $\delta=0$, and this shows that $\{\delta_i\}$ is linearly dependent.
Since every element in $D_2(1-p)$
is of the form $\gamma\otimes\alpha+\delta\otimes\alpha^\perp$, we
also have the direction (iii) $\implies$ (ii).

To complete the proof, we proceed to show (ii) $+$ (iii) $\implies$ (i). If
the subspace $D_2(1-p)$ of $M_{2\times (n-k)}$ is completely
entangled, then its orthogonal complement $H_2$ in $M_{2\times (n-k)}$
is spanned by product vectors by \cite{aug}, and so we see that the equality
holds in (\ref{dim_2}), and so in (\ref{dim}). By the condition (ii), we see that
$$
\dim\spa\{H_1,H_2\}=2n-k-\ell=\dim D^\perp.
$$
This shows that $D^\perp$ is spanned by product vectors, and completes the proof.

For the final claim, we note that $\dim D_2=\dim D_2(1-p)$ since $\{\gamma_i\}$ is linearly independent. Furthermore,
condition (iii) implies that $\dim D_2(1-p)\le (n-k)-1$ by \cite{wall}, \cite{part}, and so we have
$\dim D\le k+(n-k-1)=n-1$.
\end{proof}

In the case that $D_1$ is of the form $\beta\otimes \mathbb C^2$, we denote by $p$ the one
dimensional projection onto the vector $\beta$ then $D_2=D_2(1-p)$,
and so we see that the three conditions in Proposition \ref{decompo}
are automatically satisfied. In this case, we have $\dim D\le n$.

The final claim of Proposition \ref{decompo} on the dimensions of $D$ explains why there
are plenty of product vectors in $D^\perp$. Note that there are infinitely many product vectors for generic
$(n+1)$-dimensional subspaces of $M_{2\times n}$, whereas there are exactly $n$ product vectors for
generic $n$-dimensional subspaces of $M_{2\times n}$. See \cite{walgate}.

Note that conditions (ii) and (iii) give us convenient tests to determine if $D^\perp$ is spanned by
product vectors or not. For example, Let $D_1$ be the one dimensional subspace of $M_{2\times 3}$ spanned by
$e_{11}$. If $D_2$ is spanned by $e_{12}+e_{21}$, then $(D_1\oplus D_2)^\perp$ is not spanned by product vectors.
On the other hand, if $D_2$ is spanned by $e_{13}+e_{22}$ then $(D_1\oplus D_2)^\perp$ is spanned by
product vectors.

%%%%%%%%%%%%%%%%%%%%%%%%%%%%%%%%%%%%%%%%%%%%%%%%%%%%%%%%%%
\section{exposed faces for decomposable maps}

Now, we turn our attention to the exposedness of the face $\Phi_D$.
The space ${\mathcal L}(M_m,M_n)$ of all linear maps from $M_m$ into $M_n$ and the tensor product
$M_n\otimes M_m$ are dual each other with respect
to the bilinear pairing
$$
\langle y\otimes x, \phi \rangle ={\rm Tr}(\phi(x)y^\ttt),
$$
for $x\in M_m, y\in M_n$ and $\phi\in{\mathcal L}(M_m,M_n)$.

The dual cone of $\mathbb D$ in ${\mathcal L}(M_m,M_n)$ with respect to this pairing will be denoted by
$\mathbb T$. It is known that $\mathbb T$ consists of all positive semi-definite
block matrices $A$ in $M_n\otimes M_m=M_m(M_n)$ whose block transposes (or partial transposes) $A^\tau$
are also positive semi-definite. If $\sigma(D,E)$ is a face of $\mathbb D$, then its dual face in $\mathbb T$
is of the form
$$
\tau(D^\perp,E^\perp)=\{A\in \mathbb T: R(A)\subset D^\perp, R(A^\tau)\subset E^\perp\},
$$
where $R(A)\subset \mathbb C^n\otimes \mathbb C^m$ denotes the range space of $A$.
Therefore, we see that the face $\sigma(D,E)$ of $\mathbb D$ is
exposed if and only if there is an interior point $A$ of
$\tau(D^\perp, E^\perp)$ such that
\begin{equation}\label{face-des}
\sigma(D,E)= \{\phi\in\mathbb D: \langle A,\phi\rangle =0\}.
\end{equation}
Note that the interior of $\tau(D^\perp,E^\perp)$ consists of all
$A\in\mathbb T$ such that $R(A)= D^\perp$ and $R(A^\tau)=
E^\perp$. Elements in the cone $\mathbb T$ is called a
PPT(positive partial transpose) state if it is normalized.

The dual cone of $\mathbb P_1$ with respect to this pairing
is
$$
\mathbb V_1:=M_n^+\otimes M_m^+,
$$
which is nothing but the convex cone generated by all separable
states in $M_n\otimes M_m$ \cite{eom-kye}. We say that a face
$\sigma(D,E)$ is {\sl exposed by separable states} if there is $A\in
M_n^+\otimes M_m^+$ satisfying (\ref{face-des}). It was shown
\cite{kye_decom} that $\sigma(D,E)$ is a face of $\mathbb D$ which
is exposed by separable states if and only if the pair
$(D^\perp,E^\perp)$ satisfies the range criterion.

If $D$ is a one dimensional subspace of $M_{m\times n}$ spanned by
$V$ which is not of rank one, then Lemma 2.3 of \cite{marcin_exp}
essentially shows that the pair $(D^\perp,\{0\}^\perp)$ satisfies the range criterion.
In the case of $m=2$, it was shown in \cite{aug} that if $D$
is a completely entangled subspace of $M_{2\times n}$, then the pair
$(D^\perp,\{0\}^\perp)$ also satisfies the range criterion, which
will be crucial in our discussion. We state this result separately.

\begin{lemma} \cite{aug}\label{aug-state}
If $D$ is a completely entangled subspace of $M_{2\times n}$, then
$(D^\perp,\{0\}^\perp)$ satisfies the range criterion.
\end{lemma}

Combining with \cite{kye_decom}, we have the following:

\begin{corollary}\label{lemma-aug}
If $D$ is a completely entangled subspace of $M_{2\times n}$, then $\Phi_D$ is a face of $\mathbb D$
which is exposed by separable states.
\end{corollary}

It is apparent that $\mathbb V_1\subset\mathbb T$ in general since
$\mathbb D\subset\mathbb P_1$, and $\mathbb V_1=\mathbb T$ if and
only if $\mathbb P_1=\mathbb D$. It is of great importance in
entanglement theory to know in what circumstance elements in
$\mathbb T$ belong to $\mathbb V_1$. In the case of $m=2$,
Woronowicz \cite{woronowicz} actually showed that $\mathbb T=\mathbb
V_1$ when and only when $n\le 3$ to get the result mentioned in the
introduction.

The structures of the cone $\mathbb T$ for $m=2$ have
been also studied extensively in \cite{kraus}. Especially, it was
shown that if $A\in \mathbb T$ is supported on $\mathbb C^N \otimes
\mathbb C^2$ and has a product vector $\eta\otimes \xi$ with
$$
\eta\otimes \xi \in \ker A, \qquad \eta\otimes \xi^\perp \notin \ker A,
$$
then $A$ can be expressed by
$$
A=cq+\tilde A
$$
with a rank one projection $q$ onto a product vector in the range of $A$, a positive real number $c$
and $\tilde A\in \mathbb T$ such that
\begin{enumerate}
\item[(i)]
$\dim R(\tilde A)=\dim R(A)-1$ and $\dim R({\tilde A}^\tau)=\dim R(A^\tau)-1$,
\item[(ii)]
$\tilde A$ is supported on $\mathbb C^{N-1}\otimes \mathbb C^2$,
\item[(iii)]
$\tilde A\in\mathbb V_1$ if and only if $A\in \mathbb V_1$.
\end{enumerate}
By the proof of the above result, we also see that $\eta\otimes \xi$ and $\eta\otimes \xi^\perp$ are in the kernel of $\tilde A$.
Here, $A \in \mathbb T$ is said to be supported on $\mathbb C^N
\otimes \mathbb C^M$ if there exist an $N$-dimensional subspace $B$ of $\mathbb C^n$
and and an $M$ dimensional subspace $C$ in $\mathbb C^m$ such that
$R(A) \subset B\otimes C$ and $R(A) \subset B'\otimes C'$ implies $B\subset B'$ and $C\subset C'$

\begin{lemma}
Let $\Phi_D$ be a face of $\mathbb D$. Then $\Phi_D$ is exposed if and only if it is exposed by
separable states.
\end{lemma}

\begin{proof}
%If $D$ is completely entangled, then there is nothing to prove. We consider the case when $D$ is not completely entangled.
Let $D_1$ be the completely separable subspace of $D$ consisting of all product vector in $D$.
Then we have $\Phi_D=\sigma(D,E)$, where
\begin{equation}\label{par-con}
E=\{\eta\otimes \xi: \eta\otimes \bar\xi\in D_1\}.
\end{equation}
If $D_1=\beta \otimes \mathbb C^2$ for a vector $\beta \in \mathbb
C^n$ then $D_2$ is a completely entangled subspace of $(1-p)\mathbb
C^n \otimes \mathbb C^2$ with the projection $p \in M_n$ onto
$\beta$, so there is nothing to prove by Lemma \ref{aug-state}.

Next, we consider the case $D_1=B\otimes \alpha$ for a nonzero $\alpha \in \mathbb C^2$.
Assume that there exists an exposed face $\Phi_D=\sigma(D,E)$ which is not exposed by
separable states. Then there is $A\in\mathbb T\setminus \mathbb V_1$
which is supported on $\mathbb C^n \otimes \mathbb C^2$ and exposes
$\Phi_D$. Put $k=\dim D_1$. Note that
$$
\dim R(A)=2n-\dim D,\qquad \dim R(A^\tau)=2n-k.
$$
We can apply the above process in $k$-times, then we get $\tilde A\in \mathbb T\setminus \mathbb V_1$
such that
\begin{enumerate}
\item[(i)]
$\dim R(\tilde A)=2n-\dim D-k$ and $\dim R({\tilde A}^\tau)=2(n-k)$,
\item[(ii)]
$\tilde A$ is supported on $\mathbb C^{n-k}\otimes \mathbb C^2$.
\end{enumerate}
Let $\tilde D$ be the orthogonal complement of $R(\tilde A)$ in $\mathbb C^{n-k}\otimes \mathbb C^2$.
Then $\tilde D$ is completely entangled and $\Phi_{\tilde D}$ is a face of $\mathbb D$
which is not exposed by separable states. This is absurd by Corollary \ref{lemma-aug}.
\end{proof}

\begin{theorem}\label{main}
For a subspace $D$ of $M_{2\times n}$, the following are equivalent:
\begin{enumerate}
\item[(i)]
$\Phi_D$ is an exposed face of $\mathbb D$.
\item[(ii)]
$\Phi_D$ is a face of $\mathbb D$ which is exposed by separable states.
\item[(iii)]
The set of product vectors in $D$ forms a subspace and $D^\perp$ is spanned by product vectors.
\end{enumerate}
\end{theorem}

\begin{proof}
It remains to show the direction (iii) $\implies$ (ii). We denote by
$D_1$ the subspace of $D$ consists of all product vectors in $D$.
Then $D_1$ is completely separable subspace. We first consider the
case when $D_1$ is of the form $B\otimes \alpha $ for a fixed vector
$\alpha\in\mathbb C^2$ and a subspace $B$ of $\mathbb C^n$. We stick
to the notations in Proposition \ref{decompo} for $\{\beta_i :
i=1,\dots,l\}$ and $\{\delta_i : i=1,\dots,l\}$, for which the
conditions in Proposition \ref{decompo} are satisfied.

Take a basis $\{\xi_j:j=1,2,\dots, k\}$ of $B$. Then for each $j=1,2,\dots, k$, we can take $\eta_j\in B^\perp$
such that
$$
(\eta_j\,|\delta_i)_{B^\perp}=-(\xi_j\,|\, \beta_i)_B,\qquad i=1,2,\dots,\ell.
$$
This is possible since $\{\delta_i\}$ is linearly independent by Proposition \ref{decompo}. Then we see that
$$
(\xi_j+\eta_j\,|\,\beta_i+\delta_i)_{\mathbb C^n}=(\xi_j\,|\, \beta_i)_B+(\eta_j\,|\delta_i)_{B^\perp}=0,
$$
for each $j=1,2,\dots,k$ and $i=1,2,\dots,\ell$. Since $\{\xi_j\}$ is linearly independent, we also see that
$$
\{(\xi_j+\eta_j)\otimes \alpha^\perp: j=1,2,\dots,k\}
$$
is linearly independent set of product vectors
belongs to $D^\perp$, which does not belong to $M_{2\times n}(1-p)$.

Since $D_2(1-p)$ is completely entangled subspace of $M_{2\times n}(1-p)$,
there are $\zeta_\iota\in\mathbb C^2$ and $\omega_\iota\in B^\perp$ such that
$$
D_2(1-p)^\perp=\spa\{ \omega_\iota\otimes \zeta_\iota\},\qquad M_{2\times n}(1-p)=\spa\{\omega_\iota\otimes \bar\zeta_\iota\}
$$
by Lemma \ref{aug-state}. Therefore, if we put $E$ by (\ref{par-con}) then we have
$$
D^\perp=\spa\{(\xi_j+\eta_j)\otimes \alpha^\perp, \omega_\iota\otimes \zeta_\iota\},
\qquad E^\perp=\spa\{(\xi_j+\eta_j)\otimes \overline{ \alpha^\perp}, \omega_\iota\otimes \bar\zeta_\iota\}.
$$
Therefore, we see that the pair $(D^\perp, E^\perp)$
satisfies the range criterion, and conclude that $\Phi_D=\sigma(D, E)$ is a face of $\mathbb D$
which is exposed by separable states by \cite{kye_decom}.

The remaining case $\beta\otimes \mathbb C^2$ is obvious.
\end{proof}

\begin{corollary}
Let $(D,E)$ be a pair of subspaces in $M_{2\times n}$, where $E$ is completely separable. Then
$(D^\perp,E^\perp)$ satisfies the range criterion if and only if every product vector in $D$
is of the form $\eta\otimes \bar \xi$ for $\eta\otimes \xi\in E$ and $D^\perp$ is spanned by product vectors.
\end{corollary}

\begin{proof}
If every product vector in $D$
is of the form $\eta\otimes \bar \xi$ for $\eta\otimes \xi\in E$ and $D^\perp$ is spanned by product vectors
then $\Phi_D=\sigma(D,E)$ is a face of $\mathbb D$
which is exposed by separable states by Theorem \ref{main}. Therefore, the pair $(D^\perp,E^\perp)$
satisfies the range criterion.

Conversely, suppose that $(D^\perp,E^\perp)$ satisfies the range criterion.
Then $\sigma(D,E)$ is a face of $\mathbb D$ which is exposed by separable states. Since $E$ is
completely separable, we see that $\sigma(D,E)=\Phi_D$. Then the same argument as in
Proposition \ref{face} shows that every product vector in $D$
is of the form $\eta\otimes \bar \xi$ for $\eta\otimes \xi\in E$. It is clear that
$D^\perp$ is spanned by product vectors if $(D^\perp,E^\perp)$ satisfies the range criterion.
\end{proof}

Before finishing this note, we remark on the question when $\Phi_D$
becomes an unexposed face of $\mathbb D$. From now on, we suppose
that $\dim D\ge 2$ and the set $D_1$ of all product vectors in $D$
forms a nonzero subspace of $D$. Note that if $\dim D=1$ then
$\Phi_D$ is exposed by \cite{marcin_exp} and \cite{yopp}. We further
suppose that $D^\perp$ is not spanned by product vectors,
otherwise $\Phi_D$ becomes an exposed face by Theorem \ref{main}.

If $m=n=2$ then there is only one possibility that $\dim D=2$ and
$\dim D_1=1$ by Proposition \ref{ro2n}. If we put $E$ by
(\ref{par-con}) then $\Phi_D=\sigma(D,E)$ is nothing but the
unexposed face of $\mathbb D$ considered in Proposition 2.5 of
\cite{byeon-kye}. Note that generic $2$-dimensional subspaces of
$M_{2\times 2}$ have two product vectors, but the $2$-dimensional
space $D^\perp$ has only one product vector.

In the case of $m=2$ and $n=3$, we have three possibilities:
\begin{enumerate}
\item
$\dim D=2$ and $\dim D_1=1$.
\item
$\dim D=3$ and $\dim D_1=2$.
\item
$\dim D=3$ and $\dim D_1=1$.
\end{enumerate}
In any cases, one can show that $\Phi_D$ is a face of $\mathbb D$
by the similar direct calculations as were done in the proof of
Proposition 2.5 in \cite{byeon-kye}. So, all of these give rise to
unexposed faces of $\mathbb D$. We note that there are infinitely
many product vectors in $D^\perp$ in the first and second cases.
The product vectors in $D^\perp$ spans the
$2$-dimensional space in the second case, and $1$-dimensional space in the
third case. Both of them are not
generic cases, since generic $3$-dimensional subspaces of
$M_{2\times 3}$ have three product vectors.

We did not touch the question when $\Phi_D$ becomes a face of the
cone $\mathbb P_1$ consisting of all positive linear maps, which is
beyond the scope of this note. One necessary condition for $\Phi_D$
to be a face of $\mathbb P_1$ is that $D^\perp$ has a product vector
by \cite{kye_canad}. Numerical searches in \cite{lein} indicate that
there are $6$-dimensional completely entangled subspaces $D$ of
$M_{3\times 4}$ whose orthogonal complements are also completely
entangled. In this case, $\Phi_D$ is not contained in any maximal face
of $\mathbb P_1$, and so the interior of $\Phi_D$ lies in the interior of the whole cone
$\mathbb P_1$ even though $D$ is completely entangled.


\begin{thebibliography}{99}

\bibitem{aug}
R. Augusiak, J. Tura and M. Lewenstein, \it A note on the optimality
of decomposable entanglement witnesses and completely entangled
subspaces, \rm J. Phys. A \bf 44 \rm (2011), 212001.

\bibitem{bdmsst}  C. H. Bennett, D. P. DiVincenzo, T. Mor, P. W. Shor, J. A. Smolin and B. M. Terhal,
\it Unextendible Product Bases and Bound Entanglement,
\rm Phys. Rev. Lett. \bf 82 \rm (1999), 5385--5388.

\bibitem{bhat}
B. V. R. Bhat, \it A completely entangled subspace of maximal dimension,
\rm  Int. J. Quant. Inf. \bf 4 \rm (2006), 325--330.

\bibitem{byeon-kye}
E.-S. Byeon and S.-H. Kye,
\it Facial structures for positive linear maps in the two dimensional matrix algebra,
\rm Positivity \bf 6 \rm (2002), 369--380.

\bibitem{chen}
L. Chen and D. \v Z. Djokovi\'c,
\it Description of rank four PPT entangled states of two qutrits
\rm arXiv:1105.3142.

\bibitem{choi_kye}
H.-S. Choi and S.-H. Kye,
\it Facial Structures for Separable States,
\rm J. Korean Math. Soc., to appear.

\bibitem{eom-kye}
M.-H. Eom and S.-H. Kye,
\it Duality for positive linear maps in matrix algebras,
\rm Math. Scand. \bf 86 \rm (2000), 130--142.

\bibitem{ha+kye}
K.-C. Ha, S.-H. Kye and Y. S. Park,
\it Entanglements with positive partial transposes arising from indecomposable positive linear maps,
\rm Phys. Lett. A \bf 313 \rm (2003), 163--174.

\bibitem{p-horo}
P. Horodecki,
\it Separability criterion and inseparable mixed states with positive partial transposition,
\rm Phys. Lett. A \bf 232 \rm (1997),  \rm 333--339.

\bibitem{kraus}
B. Kraus, J. I. Cirac, S. Karnas and M. Lewenstein,
\it Separability in $2\times N$ composite quantum systems,
\rm Phys. Rev. A \bf 61, \rm (2000), 062302.

\bibitem{kye_canad}
S.-H. Kye,
\it Facial structures for positive linear maps between matrix algebras,
\rm Canad. Math. Bull. \bf 39 \rm (1996), 74--82.

\bibitem{kye-cambridge}
S.-H. Kye,
\it On the convex set of all completely positive linear maps in matrix algebras,
\rm Math. Proc. Cambridge Philos. Soc. \bf 122 \rm (1997), 45--54.

\bibitem{kye-2by2_II}
S.-H. Kye,
\it Facial structures for unital positive linear maps in the two dimensional matrix algebra,
\rm Linear Alg, Appl. \bf 362 \rm (2003), 57--73.

\bibitem{kye_decom}
S.-H. Kye,
\rm Facial structures for decomplsable positive linear maps in matrix algebras,
\rm Positivity \bf 9 \rm (2005), 63--79.

\bibitem{lein}
J. M. Leinaas, J. Myrheim and P. \O . Sollid,
\it Numerical studies of entangled PPT states in composite quantum systems,
\rm Phys. Rev. A \bf 81 \rm (2010), 062329.

\bibitem{marcin_exp}
M. Marciniak,
\it Rank properties of exposed positive maps,
\rm arXiv:1103.3497.

\bibitem{part}
K. R. Parthasarathy,
\it On the maximal dimension of a completely entangled subspace for finite level quantum systems,
\rm Proc. Indian Acad. Sci. Math. Sci. \bf 114 \rm (2004), no. 4, 365--374.

\bibitem{sko}
\L . Skowronek,
\it Three-by-three bound entanglement with general unextendible product bases,
\rm arXiv:1105.2709.

\bibitem{stormer}
E. St\o rmer,
\it Positive linear maps of operator algebras,
\rm Acta Math. \bf 110 \rm (1963), 233--278.

\bibitem{walgate}
J. Walgate and A. J. Scott, \it Generic local distinguishability and
completely entangled subspaces, \rm J. Phys. A \bf 41 \rm (2008), 375305.

\bibitem{wall}
N. R. Wallach,
\it An Unentangled Gleason's Theorem,
\rm Contemp. Math. \bf 305 \rm (2002), 291--298

\bibitem{woronowicz}
S. L. Woronowicz,
\it Positive maps of low dimensional matrix algebras,
\rm Rep. Math. Phys. \bf 10 \rm (1976), 165--183.

\bibitem{yopp}
D. A. Yopp and R. D. Hill,
\it Extremals and exposed faces of the cone of positive maps,
\rm Linear and Multilinear Algebra, \bf 53 \rm (2005), 167--174.

\end{thebibliography}
\end{document}